\documentclass[11pt]{amsart}

\usepackage[margin=1in]{geometry}
\usepackage{amsmath,amssymb,amsfonts,mathtools,bm,mathrsfs}
\usepackage{booktabs,array,longtable,tabularx,ragged2e,microtype}
\usepackage[
  backend=biber,
  style=numeric-comp,
  sorting=none,
  natbib=true,
  maxbibnames=99,
  doi=true,
  url=false,
  isbn=false
]{biblatex}
\usepackage[colorlinks=true,allcolors=blue]{hyperref}
\usepackage[nameinlink,noabbrev]{cleveref}
\usepackage{etoolbox}

\makeatletter
\newcommand{\authorinfo}[1]{\gdef\@authorinfo{#1}}
\let\@authorinfo\@empty
\patchcmd{\@setauthors}
  {\MakeUppercase{\authors}}
  {\MakeUppercase{\authors}%
   \ifx\@authorinfo\@empty\else
     \par\smallskip\normalfont\@authorinfo
   \fi}
  {}
  {\PackageWarning{main}{Could not patch amsart author block}}
\makeatother
\hypersetup{
  pdftitle={Projection, Memory, and the Validity of Scalar Early-Warning Signals: A Schur--Volterra Theory with an AMOC Application},
  pdfauthor={Mauricio Herrera-Marin},
  pdfsubject={Projection, Volterra memory, Schur-resolvent stability, and scalar early-warning signals},
  pdfkeywords={Mori--Zwanzig projection, Volterra memory, Schur complement, early-warning signals, nonnormality, AMOC}
}
\usepackage{enumitem}
\usepackage{graphicx}
\usepackage{subcaption}
\usepackage{xcolor}
\usepackage{tikz}
\usetikzlibrary{arrows.meta,positioning,fit,calc}

\newtheorem{theorem}{Theorem}[section]
\newtheorem{proposition}[theorem]{Proposition}
\newtheorem{lemma}[theorem]{Lemma}
\newtheorem{corollary}[theorem]{Corollary}
\theoremstyle{definition}
\newtheorem{definition}[theorem]{Definition}
\newtheorem{assumption}[theorem]{Assumption}

\theoremstyle{remark}
\newtheorem{remark}[theorem]{Remark}

\DeclareMathOperator{\spec}{spec}

\DeclareMathOperator{\diag}{diag}

\DeclareMathOperator{\Var}{Var}
\DeclareMathOperator{\Corr}{Corr}

\newcommand{\HH}{\mathcal H}
\newcommand{\HP}{\mathcal H_P}
\newcommand{\HQ}{\mathcal H_Q}

\newcommand{\II}{I}
\newcommand{\RR}{\mathbb R}

\newcommand{\EE}{\mathbb E}

\newcommand{\norm}[1]{\left\lVert #1\right\rVert}
\newcommand{\abs}[1]{\left\lvert #1\right\rvert}

\newcommand{\eps}{\varepsilon}
\newcommand{\FovS}{F_{\mathrm{ovS}}}
\newcommand{\dd}{\,\mathrm d}

\title[Projection, memory, and scalar EWS]
{Projection, Memory, and the Validity of Scalar Early-Warning Signals:\\
A Schur--Volterra Theory with an AMOC Application}

\author{Mauricio Herrera-Mar\'in}
\authorinfo{%
  Faculty of Engineering, Universidad del Desarrollo, Santiago, Chile\\
  \href{mailto:mherrera@udd.cl}{mherrera@udd.cl}\quad
  \href{https://orcid.org/0000-0002-9604-3077}{ORCID: 0000-0002-9604-3077}%
}


\begin{document}

\begin{abstract}
Early-warning signals are usually computed from one or a few observables, although the underlying dynamics are high dimensional.  Projection then creates a structural ambiguity: the same scalar record may reflect spectral loss of stability, hidden-state memory, nonnormal amplification, a turning point of a physical balance, or some combination of these mechanisms.  We develop a unified operator theory that makes these alternatives explicit.

For a resolved--hidden block linearization we derive exact continuous and discrete Volterra reductions, finite and diffusive Markovian liftings, and a Schur--resolvent factorization of the full stability problem.  We distinguish three non-equivalent notions of memory: realized pathwise memory, structural Schur return, and finite-horizon projected memory.  A covariance theorem gives sufficient conditions under which a scalar observable genuinely inherits critical slowing down from a simple real multiplier approaching the unit circle.  The theorem isolates the required ingredients: spectral separation, bounded eigenvector conditioning, observability of the critical mode, stochastic excitation of that mode, and quasi-stationary forcing.  We then prove failure mechanisms.  A critical mode can be invisible or unexcited; an oscillatory crossing need not produce lag-one correlation approaching one; a nonnormal family can generate diverging scalar variance with a fixed stable spectrum; and a turning point of a physics-based scalar observable need not coincide with a spectral boundary.  Finite-horizon observations also fail to identify long memory uniquely.

The theory is illustrated with a $4400$-year Community Earth System Model experiment tracing an Atlantic Meridional Overturning Circulation hysteresis loop.  The first three vertical EOFs represent amplitude, vertical displacement, and shape deformation.  EOF1 alone detects altered susceptibility, but EOF1--EOF2 is the smallest resolved space that robustly reveals a changed Schur return and projected memory after transport recovery.  The physics-based freshwater-transport indicator $\FovS$ reaches its spline-estimated minimum in model year $1732$, $26$ years before collapse, while the modal stability margin does not decrease consistently toward zero.  Thus the scalar balance indicator and the projected operator diagnose complementary mechanisms rather than a common universal threshold.  The application is mechanistic and single-trajectory; no fractional AMOC law or validated operatorial early-warning signal is claimed.
\end{abstract}

\maketitle

\section{Introduction}

A scalar early-warning signal is a statement about a projection.  The climate system, an ecosystem, a fluid flow, a molecular assembly, or a network may evolve in a state space of very large dimension, yet the practical diagnostic is often built from a single measured quantity.  Increasing variance and lag-one autocorrelation are then interpreted through critical slowing down: a dominant restoring rate approaches zero, or a discrete multiplier approaches the unit circle, so perturbations decay more slowly and stationary fluctuations increase \citep{Scheffer2009,Dakos2012,Lenton2012,Kuehn2011,Kuehn2015}.  The interpretation is mathematically sound in the regime for which it was derived.  It is not automatically inherited by an arbitrary observable of a partially observed, multiscale, nonnormal, or nonstationary system.

Projection is the central difficulty.  The Mori--Zwanzig identity shows that eliminating unresolved degrees of freedom replaces a Markovian full-state evolution by a reduced equation containing an instantaneous term, a memory term, and orthogonal dynamics \citep{Mori1965,Zwanzig1973,ChorinHald2013,ChorinStinis2006}.  The reduced observable is therefore generally not a state.  Its variance and autocorrelation may respond to hidden relaxation, colored forcing, unresolved initial conditions, or nonnormal input--response geometry.  Data-driven closure methods explicitly estimate these memory corrections \citep{Kondrashov2015,Gouasmi2017,ParishDuraisamy2017,PanDuraisamy2018,LinLu2021,LinEtAl2022,LinEtAl2023,Gupta2025}, while Volterra theory and realization theory provide the corresponding analytical language \citep{Gripenberg1990,Pruss1993,Schilling2012}.

A second ambiguity is that early-warning-like statistics do not uniquely identify spectral criticality.  Correlated and nonstationary noise can alter rolling indicators \citep{DitlevsenJohnsen2010,BoettnerBoers2022}; oscillatory modes can obscure scalar AR(1) signatures; and nonnormal dynamics can produce strong transient or noise-driven amplification with a stable spectrum \citep{TrefethenEmbree2005,TroudeSornette2025,TroudeEtAl2026}.  Multivariate local models can recover information unavailable to one-dimensional AR(1) diagnostics, especially near oscillatory or multiscale transitions \citep{Hobden2026}.  Event-time hazard decompositions provide another example: history-dependent clustering can be informative even when it is not a bifurcation EWS \citep{HerreraMarin2026Drought}.

The purpose of this paper is not to reject scalar indicators.  It is to answer a sharper question:
\begin{quote}
\emph{Under what operator-theoretic conditions does a scalar early-warning signal survive projection, and what distinct mechanisms can imitate or invalidate that interpretation?}
\end{quote}

The answer requires connecting objects that are often treated separately:
\[
\begin{aligned}
\text{projection}
&\longrightarrow \text{Volterra memory}
\longrightarrow \text{Markovian lifting}\\
&\longrightarrow \text{Schur--resolvent stability}
\longrightarrow \text{nonnormal transfer}
\longrightarrow \text{scalar statistics}.
\end{aligned}
\]
The resulting framework separates a physics-based turning point from a spectral boundary, and separates asymptotic stability from finite-horizon receptivity.

\subsection{Main contributions}

The main results are the following.
\begin{enumerate}[label=(\roman*),leftmargin=0.8cm]
\item We derive exact continuous and discrete Volterra equations generated by elimination of a hidden block, including the unresolved-initial-condition and hidden-forcing terms that are frequently merged into an undifferentiated ``memory'' residual.
\item We give exact finite Markovian realizations for rational kernels and diffusive realizations for completely monotone kernels.  Fractional and tempered-fractional kernels appear as structured special cases, not as generic consequences of projection.
\item A Schur--resolvent factorization yields an exact full-state stability criterion and identifies the memory-renormalized stationary boundary
\[
\II-A-B(\II-D)^{-1}C.
\]
\item We define structural memory, finite-horizon projected memory, hidden-to-observed transfer, and the resolved Schur resolvent.  Their values and invariance properties are made explicit.
\item We prove a scalar critical-slowing-down theorem.  If a simple real multiplier $r\uparrow1$ is spectrally isolated, uniformly conditioned, observed by the scalar, and excited by the noise, then
\[
\Var(y)\sim \frac{C_y}{1-r^2},
\qquad
\Corr(y_{n+1},y_n)\to1.
\]
\item We prove counterexamples showing that each assumption matters: invisible or unexcited critical modes, oscillatory crossings, nonnormal variance inflation at fixed spectral radius, and scalar turning points without spectral criticality.
\item We show that any finite memory horizon is non-identifying: two stable hidden realizations can agree through an arbitrary number of kernel coefficients and differ only later.
\item We use a fully audited AMOC hysteresis calculation to distinguish scalar transport, vertical memory, nonnormal hidden-to-observed transfer, and a physics-based freshwater-balance indicator.
\end{enumerate}

\subsection{Scope}

The mathematical theory is linear because the local relation among projection, memory, Schur complements, covariance, and nonnormality is already nontrivial and can be stated exactly.  For a nonlinear slowly forced system, the matrices below are interpreted as local linearizations or local regression operators.  The AMOC application does not establish an observational collapse forecast, a fractional climate law, or a universal operatorial EWS.  Its purpose is to show that the distinctions in the theory occur in a state-of-the-art climate-model trajectory.

\begin{figure}[t]
\centering
\resizebox{\textwidth}{!}{%
\begin{tikzpicture}[
  node distance=12mm and 18mm,
  box/.style={draw,rounded corners,align=center,minimum height=9mm,minimum width=27mm,fill=blue!4},
  arr/.style={-{Latex[length=2mm]},thick},
  note/.style={align=center,font=\small}
]
\node[box] (full) {full Markov state\\$x=(p,q)$};
\node[box,right=of full] (proj) {resolved observable\\$p=Px$};
\node[box,right=of proj] (vol) {Volterra reduction\\$A+\{K_j\}$};
\node[box,right=of vol] (stat) {scalar statistic\\variance, AC(1), turn};
\draw[arr] (full) -- node[above,font=\small]{projection} (proj);
\draw[arr] (proj) -- node[above,font=\small]{elimination} (vol);
\draw[arr] (vol) -- node[above,font=\small]{observation} (stat);
\node[box,below=of vol] (lift) {Markovian lift\\hidden realization};
\node[box,below=of stat] (schur) {Schur resolvent\\and finite-time gain};
\draw[arr] (vol) -- node[right,font=\small]{realization} (lift);
\draw[arr] (lift) -- (schur);
\draw[arr] (schur) -- node[right,font=\small]{sufficient conditions} (stat);
\end{tikzpicture}%
}
\par\smallskip
{\small\centering Projection loses information; lifting restores Markovianity for the selected kernel, not the original hidden state.\par}
\caption{Logical structure.  Scalar early-warning interpretation is valid only after the projected memory, hidden-state observability, spectral geometry, noise excitation, and forcing time scale have been controlled.}
\label{fig:schematic}
\end{figure}

\section{Projection produces an exact Volterra equation}
\label{sec:projection}

Let
\[
\HH=\HP\oplus\HQ
\]
be a finite-dimensional real or complex Hilbert space, with orthogonal projectors $P$ and $Q=\II-P$.  Write $x=(p,q)$ with $p\in\HP$ resolved and $q\in\HQ$ hidden.

\subsection{Continuous time}

Consider
\begin{equation}
\label{eq:continuous-full}
\dot x(t)=Lx(t)+f(t),
\qquad
L=\begin{pmatrix}A&B\\ C&D\end{pmatrix},
\qquad
f=\begin{pmatrix}f_P\\ f_Q\end{pmatrix}.
\end{equation}

\begin{theorem}[Exact continuous elimination]
\label{thm:continuous-elimination}
Assume that $D$ generates a strongly continuous semigroup on $\HQ$.  Every mild solution of \eqref{eq:continuous-full} satisfies
\begin{align}
\dot p(t)
={}&Ap(t)+\int_0^tK(t-s)p(s)\dd s+\eta(t),
\label{eq:continuous-reduced}\\
K(t)={}&Be^{tD}C,
\label{eq:continuous-kernel}\\
\eta(t)={}&f_P(t)+Be^{tD}q_0+
\int_0^tBe^{(t-s)D}f_Q(s)\dd s.
\label{eq:continuous-orthogonal}
\end{align}
Conversely, if $p$ solves \eqref{eq:continuous-reduced} and
\[
q(t)=e^{tD}q_0+
\int_0^te^{(t-s)D}\bigl(Cp(s)+f_Q(s)\bigr)\dd s,
\]
then $(p,q)$ solves \eqref{eq:continuous-full}.
\end{theorem}

\begin{proof}
Variation of constants in the hidden equation gives the displayed expression for $q$.  Substitution into $\dot p=Ap+Bq+f_P$ yields \eqref{eq:continuous-reduced}--\eqref{eq:continuous-orthogonal}; the converse is direct.
\end{proof}

The kernel $K(t)$ is structural.  The term $Be^{tD}q_0$ contains the unresolved initial condition, while the last convolution in \eqref{eq:continuous-orthogonal} transports hidden forcing.  These terms are not interchangeable.

\subsection{Discrete time}

The sampled or intrinsically discrete system is
\begin{equation}
\label{eq:discrete-full}
x_{n+1}=Jx_n+u_n,
\qquad
J=\begin{pmatrix}A&B\\ C&D\end{pmatrix}.
\end{equation}

\begin{theorem}[Exact discrete elimination]
\label{thm:discrete-elimination}
For every $n\ge0$,
\begin{align}
q_n={}&D^nq_0+
\sum_{j=0}^{n-1}D^{n-1-j}\bigl(Cp_j+u_{Q,j}\bigr),
\label{eq:q-elim}\\
p_{n+1}={}&Ap_n+
\sum_{\ell=1}^{n}K_\ell p_{n-\ell}+\eta_n,
\label{eq:discrete-volterra}
\end{align}
where
\begin{equation}
\label{eq:discrete-kernel}
K_\ell=BD^{\ell-1}C,
\qquad \ell\ge1,
\end{equation}
and
\[
\eta_n=u_{P,n}+BD^nq_0+
\sum_{j=0}^{n-1}BD^{n-1-j}u_{Q,j}.
\]
Thus every finite-dimensional hidden block generates an exact matrix-valued Volterra equation.
\end{theorem}

\begin{proof}
Equation \eqref{eq:q-elim} follows by induction.  Insert it into $p_{n+1}=Ap_n+Bq_n+u_{P,n}$ and reindex.
\end{proof}

\subsection{Transfer function and three notions of memory}

For $z\notin\spec(D)$ define
\begin{equation}
\label{eq:transfer}
\widehat K(z)=B(z\II-D)^{-1}C.
\end{equation}
The coefficients in the expansion at infinity are exactly the memory kernels:
\[
\widehat K(z)=\sum_{\ell\ge1}z^{-\ell}K_\ell.
\]
This suggests three distinct notions.

\begin{definition}[Realized, structural, and dynamical memory]
\label{def:memory-types}
\begin{enumerate}[label=(\alph*)]
\item A \emph{realized memory contribution} is the path-dependent value
\[
\mathfrak m_n^{\mathrm{real}}
=\sum_{\ell=1}^{n}K_\ell p_{n-\ell}.
\]
\item If $\rho(D)<1$, the \emph{structural static return} is
\begin{equation}
\label{eq:static-return}
\mathcal N_0=B(\II-D)^{-1}C=\sum_{\ell\ge1}K_\ell.
\end{equation}
\item The \emph{finite-horizon projected-memory correction} is
\begin{equation}
\label{eq:dynamic-memory}
\mathcal M_h=PJ^hP-A^h.
\end{equation}
\end{enumerate}
\end{definition}

The first object depends on the realized history, the second on the hidden realization, and the third on excursions that leave and return to the resolved subspace within $h$ steps.  They should not be identified with one another.  A positive scalar Hawkes decomposition is a useful special case: the history term $\alpha\sum_{k\ge1}g_ke_{n-k}$ is realized Volterra memory, while the kernel and its lift encode structural memory \citep{HerreraMarin2026Drought}.

\section{Markovian lifting and fractional memory}
\label{sec:lifting}

Projection produces memory.  Lifting performs the reverse operation only for a selected reduced kernel: it restores Markovianity by adding auxiliary variables, but it does not reconstruct the original hidden state.

\subsection{Finite rational realizations}

\begin{theorem}[Exact finite lifting]
\label{thm:finite-lift}
Suppose
\[
K(t)=\sum_{m=1}^MB_me^{tD_m}C_m.
\]
Define
\[
z_m(t)=\int_0^te^{(t-s)D_m}C_mp(s)\dd s.
\]
Then the Volterra equation
\[
\dot p=Ap+\int_0^tK(t-s)p(s)\dd s+f
\]
is equivalent in its resolved component to
\begin{equation}
\label{eq:finite-lift}
\dot p=Ap+\sum_{m=1}^MB_mz_m+f,
\qquad
\dot z_m=D_mz_m+C_mp,
\qquad z_m(0)=0.
\end{equation}
\end{theorem}

\begin{proof}
Differentiate the definition of $z_m$ and use $\sum_mB_mz_m=(K*p)(t)$.
\end{proof}

The discrete system \eqref{eq:discrete-full} is already such a realization.  Hidden similarity transformations
\[
(B,C,D)\mapsto(BT^{-1},TC,TDT^{-1})
\]
leave every coefficient $BD^{\ell-1}C$ unchanged.  Hence the input--output memory does not identify hidden coordinates uniquely.

\subsection{Completely monotone and fractional kernels}

Let $k:(0,\infty)\to[0,\infty)$ be completely monotone and locally integrable.  Bernstein's theorem yields a positive measure $\mu$ such that
\begin{equation}
\label{eq:bernstein}
k(t)=\int_0^\infty e^{-\lambda t}\,\mu(\dd\lambda).
\end{equation}

\begin{theorem}[Diffusive lifting]
\label{thm:diffusive-lift}
For $K(t)=Bk(t)C$, assume the displayed Bochner integrals are finite.  The continuum system
\begin{align}
\dot p(t)&=Ap(t)+B\int_0^\infty z(\lambda,t)\mu(\dd\lambda)+f(t),\\
\partial_tz(\lambda,t)&=-\lambda z(\lambda,t)+Cp(t),
\qquad z(\lambda,0)=0,
\end{align}
is an exact Markovian realization of
\[
\dot p=Ap+B\int_0^tk(t-s)Cp(s)\dd s+f.
\]
\end{theorem}

\begin{proof}
Solve the linear equation for $z(\lambda,t)$ and apply Tonelli's theorem to exchange the time and $\lambda$ integrals.  The inner integral is \eqref{eq:bernstein}.
\end{proof}

\begin{corollary}[Fractional diffusive representation]
\label{cor:fractional}
For $0<\alpha<1$,
\[
k_\alpha(t)=\frac{t^{-\alpha}}{\Gamma(1-\alpha)}
=\frac{\sin(\pi\alpha)}{\pi}
\int_0^\infty e^{-\lambda t}\lambda^{\alpha-1}\dd\lambda.
\]
Thus fractional Volterra memory is an infinite positive Markovian lift.  Tempering by $e^{-\tau t}$ shifts the relaxation spectrum from $\lambda$ to $\lambda+\tau$.
\end{corollary}

\begin{remark}[What fractional memory does and does not mean]
A fractional kernel encodes a particular relaxation spectrum.  Projection alone implies a Volterra kernel, not a power law.  Fractionality must be justified by kernel identification, asymptotics, or a physical mechanism.  The AMOC application below does not make such a claim.
\end{remark}

\subsection{Approximation by exponential sums}

\begin{theorem}[Kernel-to-solution perturbation]
\label{thm:kernel-perturbation}
Let $p$ and $p_M$ solve the same linear Volterra initial-value problem on $[0,T]$ with kernels $K$ and $K_M$.  Put
\[
\eps_M=\norm{K-K_M}_{L^1(0,T)},
\qquad
M_T=\sup_{0\le t\le T}\norm{p_M(t)}.
\]
Then
\[
\sup_{0\le t\le T}\norm{p(t)-p_M(t)}
\le
T\eps_MM_T
\exp\!\left[T\bigl(\norm A+\norm K_{L^1(0,T)}\bigr)\right].
\]
\end{theorem}

\begin{proof}
For $e=p-p_M$,
\[
\dot e=Ae+K*e+(K-K_M)*p_M.
\]
Integrate, use Fubini and the $L^1$ kernel bound, and apply Gronwall's inequality.
\end{proof}

\section{Schur--Volterra stability and finite-horizon response}
\label{sec:schur}

The full operator and the reduced memory are two representations of the same dynamics.  Their spectra are connected by a Schur complement of the resolvent.

\subsection{Exact factorization}

\begin{theorem}[Schur--resolvent factorization]
\label{thm:schur-factor}
Let
\[
J=\begin{pmatrix}A&B\\C&D\end{pmatrix}
\]
act on $\HP\oplus\HQ$.  For $z\in\rho(D)$ define
\begin{equation}
\label{eq:Phi}
\Phi(z)=z\II_{\HP}-A-B(z\II_{\HQ}-D)^{-1}C.
\end{equation}
Then $z\II-J$ is invertible if and only if $\Phi(z)$ is invertible.  In finite dimensions,
\begin{equation}
\label{eq:det-factor}
\det(z\II-J)=\det(z\II-D)\det\Phi(z).
\end{equation}
When both factors are invertible,
\begin{align}
(z\II-J)^{-1}_{PP}&=\Phi(z)^{-1},
\label{eq:resolved-resolvent}\\
(z\II-J)^{-1}_{PQ}&=\Phi(z)^{-1}B(z\II-D)^{-1}.
\end{align}
\end{theorem}

\begin{proof}
Apply block Gaussian elimination to $z\II-J$.  The $PP$ Schur complement is \eqref{eq:Phi}; inversion of the triangular factors gives the stated blocks.
\end{proof}

\begin{corollary}[Exact Schur stability criterion]
\label{cor:schur-stability}
Assume $\rho(D)<1$.  Then $J$ is Schur stable if and only if $\Phi(z)$ is invertible for every $|z|\ge1$.
\end{corollary}

\begin{proof}
The hidden factor has no zeros for $|z|\ge1$.  The conclusion follows from \eqref{eq:det-factor}.
\end{proof}

At the stationary boundary $z=1$,
\begin{equation}
\label{eq:static-schur}
\Phi(1)=\II-A-B(\II-D)^{-1}C.
\end{equation}
Hence a unit multiplier is not governed by $\II-A$ alone; the hidden return renormalizes the restoring operator.

\begin{definition}[Resolved boundary response]
If $J$ and $D$ are Schur stable, define
\begin{equation}
\label{eq:resolvent-response}
\mathcal R_\Phi=
\max_{|z|=1}\norm{\Phi(z)^{-1}},
\qquad
 d_\Phi=
\min_{|z|=1}\sigma_{\min}(\Phi(z)).
\end{equation}
In the Euclidean metric, $\mathcal R_\Phi=d_\Phi^{-1}$.
\end{definition}

A divergence of $\mathcal R_\Phi$ indicates an approaching unit-circle singularity of the resolved Schur factor.  Large values can also reflect poor conditioning and therefore need not reduce to a single eigenvalue distance.

\subsection{Resolved propagation and hidden transfer}

Set
\[
X_h=PJ^hP\big|_{\HP},
\qquad
Y_h=QJ^hP\big|_{\HP}.
\]

\begin{theorem}[Exact recurrence for the resolved propagator]
\label{thm:resolved-recurrence}
The blocks satisfy $X_0=\II$, $Y_0=0$ and
\begin{equation}
\label{eq:resolved-recurrence}
X_{h+1}=AX_h+
\sum_{j=0}^{h-1}BD^{h-1-j}CX_j.
\end{equation}
Consequently, $X_h-A^h$ is exactly the sum of paths that leave the resolved space and return within $h$ steps.
\end{theorem}

\begin{proof}
Block multiplication gives $X_{h+1}=AX_h+BY_h$ and $Y_{h+1}=CX_h+DY_h$.  Solve the second recurrence and substitute into the first.
\end{proof}

\begin{definition}[Finite-horizon operator diagnostics]
For $H\ge1$ define
\begin{align}
\mathcal G_H&=\max_{1\le h\le H}\norm{J^h},\\
\mathcal E_H&=\max_{1\le h\le H}
\bigl(\norm{J^h}-\rho(J)^h\bigr),\\
\mathcal H_H&=\max_{1\le h\le H}\norm{PJ^hQ},\\
\mathcal M_H&=\max_{1\le h\le H}\norm{PJ^hP-A^h}.
\end{align}
They measure total gain, excess over the modal lower bound, hidden-to-observed transfer, and projected-memory correction, respectively.
\end{definition}

These quantities answer different questions from $\rho(J)$.  A Schur-stable system may have large $\mathcal G_H$, $\mathcal H_H$, or $\mathcal M_H$.

\begin{theorem}[Decoder--oracle identity]
\label{thm:decoder-oracle}
Consider two forecasts generated by the same $J$, forcing, and resolved initial condition, but initialized with hidden states $q_0$ and $\widehat q_0$.  Then
\begin{equation}
\label{eq:decoder-oracle}
\widehat p_h^{\,\mathrm{dec}}-
\widehat p_h^{\,\mathrm{orc}}
=PJ^hQ(\widehat q_0-q_0).
\end{equation}
In particular,
\[
\norm{\widehat p_h^{\,\mathrm{dec}}-
\widehat p_h^{\,\mathrm{orc}}}
\le \norm{PJ^hQ}\,\norm{\widehat q_0-q_0}.
\]
\end{theorem}

\begin{proof}
Subtract the two full-state recursions.  The initial difference lies in $\HQ$ and evolves as $J^hQ(\widehat q_0-q_0)$.  Project with $P$.
\end{proof}

\subsection{Metric covariance}

Let $G\succ0$ define $\norm{x}_G^2=x^*Gx$ and set $R=G^{1/2}$.

\begin{theorem}[Metric representation]
\label{thm:metric}
For every $h\ge0$,
\[
\norm{J^h}_G=\norm{RJ^hR^{-1}}_2.
\]
The spectrum is unchanged by the similarity $J\mapsto RJR^{-1}$.  For two metrics $G_1,G_2$ and $S=R_2R_1^{-1}$,
\[
\kappa_2(S)^{-1}\norm{J^h}_{G_1}
\le
\norm{J^h}_{G_2}
\le
\kappa_2(S)\norm{J^h}_{G_1}.
\]
The corresponding bounds hold for physically matched hidden-to-observed maps.
\end{theorem}

\begin{proof}
Change variables with $R$.  The comparison follows from
$R_2J^hR_2^{-1}=S(R_1J^hR_1^{-1})S^{-1}$.
\end{proof}

Thus spectral stability is coordinate invariant, whereas transient gain is a statement about a specified physical norm.  Robust claims should report both.

\section{When does a scalar early-warning signal survive projection?}
\label{sec:ews-theory}

We now state a sufficient theorem for genuine scalar critical slowing down.  The discrete setting matches annual local operators; the continuous analogue follows by replacing powers with semigroups and $1-r$ with the leading decay rate.

Let
\begin{equation}
\label{eq:stochastic-linear}
x_{n+1}=J_\eps x_n+\xi_{n+1},
\qquad
\EE\xi_n=0,
\qquad
\EE\xi_n\xi_n^*=\Sigma_\eps,
\end{equation}
where $J_\eps$ is Schur stable.  The scalar observation is
\[
y_n=c_\eps^*x_n.
\]

\begin{assumption}[Simple observable critical mode]
\label{ass:critical-mode}
As $\eps\downarrow0$:
\begin{enumerate}[label=(\alph*)]
\item $J_\eps$ has a simple real eigenvalue $r_\eps\in(0,1)$ with $r_\eps\to1$;
\item its spectral projector is $\Pi_\eps=v_\eps w_\eps^*$, normalized by $w_\eps^*v_\eps=1$;
\item $J_\eps=r_\eps\Pi_\eps+N_\eps$, with $\Pi_\eps N_\eps=N_\eps\Pi_\eps=0$ and
\[
\norm{N_\eps^k}\le C\gamma^k,
\qquad 0<\gamma<1,
\]
uniformly in $\eps$;
\item $\norm{v_\eps}$, $\norm{w_\eps}$, and $\norm{\Sigma_\eps}$ are uniformly bounded;
\item the mode is observed and excited:
\[
\abs{c_\eps^*v_\eps}\ge c_0>0,
\qquad
w_\eps^*\Sigma_\eps w_\eps\ge q_0>0.
\]
\end{enumerate}
\end{assumption}

\begin{theorem}[Scalar critical slowing down under projection]
\label{thm:scalar-ews}
Under \cref{ass:critical-mode}, the stationary covariance $\Gamma_\eps$ exists and
\begin{equation}
\label{eq:cov-asymptotic}
\Gamma_\eps=
\frac{q_\eps}{1-r_\eps^2}
 v_\eps v_\eps^*+R_\eps,
\qquad
q_\eps=w_\eps^*\Sigma_\eps w_\eps,
\qquad
\sup_\eps\norm{R_\eps}<\infty.
\end{equation}
Consequently,
\begin{align}
\Var(y_n)
&=
\frac{q_\eps\abs{c_\eps^*v_\eps}^2}{1-r_\eps^2}+O(1),
\label{eq:var-asymptotic}\\
\Corr(y_{n+1},y_n)
&=r_\eps+O(1-r_\eps^2)
\longrightarrow1.
\label{eq:ac-asymptotic}
\end{align}
\end{theorem}

\begin{proof}
The stationary covariance is
\[
\Gamma_\eps=\sum_{k\ge0}J_\eps^k\Sigma_\eps(J_\eps^*)^k.
\]
Since $J_\eps^k=r_\eps^k\Pi_\eps+N_\eps^k$, the projector--projector term equals
\[
\sum_{k\ge0}r_\eps^{2k}\Pi_\eps\Sigma_\eps\Pi_\eps^*
=
\frac{q_\eps}{1-r_\eps^2}v_\eps v_\eps^*.
\]
The two cross terms are uniformly bounded by a geometric series in $r_\eps\gamma$, and the $N_\eps$--$N_\eps$ term by a geometric series in $\gamma^2$.  This proves \eqref{eq:cov-asymptotic} and \eqref{eq:var-asymptotic}.  The lag-one covariance has leading term
$r_\eps q_\eps\abs{c_\eps^*v_\eps}^2/(1-r_\eps^2)$ and a uniformly bounded remainder; division by \eqref{eq:var-asymptotic} gives \eqref{eq:ac-asymptotic}.
\end{proof}

\begin{remark}[Meaning of the assumptions]
The theorem requires more than $\rho(J_\eps)\to1$.  The critical mode must be simple, separated from faster modes, not rendered ill conditioned by a diverging eigenbasis, visible in the chosen scalar, and excited by the stochastic forcing.  These are precisely the assumptions that projection can violate.
\end{remark}

\subsection{Oscillatory critical modes}

\begin{proposition}[Lag-one correlation at an oscillatory crossing]
\label{prop:oscillatory}
Consider the normal two-dimensional block
\[
J_r=r
\begin{pmatrix}
\cos\omega&-\sin\omega\\
\sin\omega&\cos\omega
\end{pmatrix},
\qquad 0<r<1,
\]
with isotropic noise and any nonzero scalar coordinate observation.  Then stationary variance diverges as $(1-r^2)^{-1}$, but
\[
\Corr(y_{n+1},y_n)=r\cos\omega
\longrightarrow\cos\omega.
\]
Unless $\omega=0$, lag-one correlation does not approach one.
\end{proposition}

\begin{proof}
The stationary covariance is a scalar multiple of the identity.  The one-step covariance in any coordinate is therefore $r\cos\omega$ times the variance.
\end{proof}

A scalar AC(1) criterion is consequently tuned to a real positive multiplier.  A Hopf-type or multiscale transition calls for multivariate or frequency-aware diagnostics.

\subsection{Slow forcing}

The theorem is local and stationary.  In a slowly forced system $J(\lambda_n)$, a rolling window is interpretable quasi-statically only when parameter drift across the window is small relative to the local relaxation and spectral-separation scales.  Near criticality the relaxation time diverges, so the quasi-stationary window requirement becomes increasingly difficult.  This is one reason that estimated tipping times are more fragile than the existence of a local trend \citep{BenYami2024}.

\section{Failure mechanisms and non-identifiability}
\label{sec:failures}

\subsection{Invisible and unexcited critical modes}

\begin{proposition}[Projection can erase critical slowing down]
\label{prop:invisible}
In \cref{thm:scalar-ews}, if $c_\eps^*v_\eps=0$, the leading divergent covariance component is absent from $y_n$.  If $w_\eps^*\Sigma_\eps w_\eps=0$, the critical mode is not stochastically excited.  In either case, scalar variance need not diverge although $r_\eps\to1$.
\end{proposition}

\begin{proof}
Both statements follow directly from the coefficient of the singular term in \eqref{eq:var-asymptotic}.  Explicit diagonal examples show that the remaining variance can stay bounded.
\end{proof}

\subsection{Nonnormal variance inflation without spectral criticality}

\begin{theorem}[A stable nonnormal false positive]
\label{thm:nonnormal-counterexample}
Fix $a,b\in(-1,1)$ with $a\ne b$ and define
\[
J_\kappa=
\begin{pmatrix}
a&\kappa\\0&b
\end{pmatrix},
\qquad
\Sigma=
\begin{pmatrix}0&0\\0&1\end{pmatrix},
\qquad
y_n=e_1^*x_n.
\]
Then $\rho(J_\kappa)=\max\{|a|,|b|\}<1$ for every $\kappa$, while the stationary scalar variance is
\begin{equation}
\label{eq:nonnormal-var}
\Var(y_n)=
\frac{\kappa^2}{(a-b)^2}
\left(
\frac{1}{1-a^2}
+
\frac{1}{1-b^2}
-
\frac{2}{1-ab}
\right).
\end{equation}
Hence $\Var(y_n)\to\infty$ as $|\kappa|\to\infty$ although the spectrum remains fixed and uniformly stable.
\end{theorem}

\begin{proof}
For $k\ge0$,
\[
(J_\kappa^k)_{12}
=\kappa\frac{a^k-b^k}{a-b}.
\]
Only noise in the second coordinate reaches $y$.  Summing the squared impulse response over $k$ gives \eqref{eq:nonnormal-var}.
\end{proof}

This counterexample separates spectral criticality from input--output amplification.  In a projected system the effective coupling $\kappa$ can change because hidden variables rotate or because the physical metric changes, even when eigenvalues remain stable.

\subsection{A physical turning point is not a spectral boundary}

\begin{proposition}[Turning-point non-equivalence]
\label{prop:turning-point}
Let $J$ be any fixed Schur-stable matrix, let $x_{n+1}=Jx_n+\xi_{n+1}$, and define a scalar observable
\[
s_n=m(\lambda_n)+c^*x_n,
\]
where $m$ is differentiable and $m'(\lambda_*)=0$.  Then the mean scalar response has a turning point at $\lambda_*$, while the spectral stability of $J$ is unchanged.  Therefore
\[
\frac{\dd}{\dd\lambda}\EE s=0
\quad\not\Rightarrow\quad
\rho(J)=1.
\]
Conversely, there are families $J_\lambda$ with $\rho(J_\lambda)\uparrow1$ for which the derivative of the observed mean is nonzero.  Hence neither implication holds without additional structural assumptions.
\end{proposition}

\begin{proof}
For the fixed stable matrix, stationarity gives $\EE x_n=0$ after centering, hence $\EE s_n=m(\lambda_n)$ and the first implication fails at any critical point of $m$.  Conversely, let $J_\lambda=\diag(\lambda,a)$ with $|a|<1$ and $\lambda\uparrow1$, and take $m(\lambda)=\lambda$ with $c=e_2$.  Then the full system approaches a spectral boundary while $\dd\EE s/\dd\lambda=1$ and the observed coordinate remains uniformly stable.  Thus the reverse implication also fails.
\end{proof}

A physics-based balance indicator may be extremely useful precisely because it is tied to a conservation or feedback mechanism.  That does not make its extremum a universal spectral threshold of every projected local operator.

\subsection{Finite-horizon memory is non-identifying}

\begin{proposition}[Arbitrarily delayed hidden return]
\label{prop:delayed-return}
For every integer $H\ge1$ and scalar $\beta\ne0$, there exists a nilpotent, hence Schur-stable, hidden realization $(B,C,D)$ such that
\[
BD^{\ell-1}C=0,
\qquad 1\le\ell\le H,
\]
but
\[
BD^HC=\beta.
\]
Thus two reduced systems can have identical memory coefficients through lag $H$ and differ at lag $H+1$.
\end{proposition}

\begin{proof}
Let $\HQ=\RR^{H+1}$ with basis $e_1,\dots,e_{H+1}$.  Set $C(1)=e_1$, $De_j=e_{j+1}$ for $j\le H$, $De_{H+1}=0$, and $Bz=\beta z_{H+1}$.  Then $D^{\ell-1}e_1=e_\ell$, which is annihilated by $B$ for $\ell\le H$, while $BD^He_1=\beta$.
\end{proof}

No finite stack of lags identifies the full hidden memory without a model class or truncation assumption.  This is an operator version of the broader principle that a projected record need not identify the mechanism producing its temporal organization.

\begin{table}[t]
\centering
\caption{Distinct mechanisms that can produce or suppress scalar warning statistics.}
\label{tab:mechanisms}
\begin{tabularx}{\textwidth}{>{\RaggedRight\arraybackslash}p{0.22\textwidth}>{\RaggedRight\arraybackslash}p{0.29\textwidth}>{\RaggedRight\arraybackslash}X}
\toprule
Mechanism & Mathematical object & Scalar consequence \\
\midrule
Real spectral criticality & simple $r\uparrow1$ with observable, excited mode & variance diverges; AC(1)$\to1$ for a real positive mode \\
Oscillatory crossing & pair $re^{\pm i\omega}$ & variance may diverge; AC(1)$\to\cos\omega$ \\
Hidden critical mode & $c^*v=0$ & no leading scalar divergence \\
Unexcited critical mode & $w^*\Sigma w=0$ & no leading stochastic divergence \\
Nonnormal amplification & large $\norm{J^h}$ or $\norm{PJ^hQ}$ at fixed spectrum & variance and finite-time response can increase without a bifurcation \\
Physical balance turning & $\dd \EE s/\dd\lambda=0$ & useful mechanistic threshold, not generally $\rho(J)=1$ \\
Projected memory & $BD^{k}C$, $PJ^hP-A^h$ & history dependence and scalar persistence need not be one-to-one \\
\bottomrule
\end{tabularx}
\end{table}

\section{Causal local inference and audit principles}
\label{sec:inference}

For a slowly evolving system, let $z_t\in\RR^d$ denote frozen coordinates and $F_t$ a prescribed forcing.  A local trailing-window model is
\begin{equation}
\label{eq:local-model}
z_{t+1}-z_t=L_tz_t+c_tF_t+b_t+\eps_{t+1},
\qquad M_t=\II+L_t.
\end{equation}
Only past observations are used.  The ridge parameter is selected in a disjoint control period, and the representation, windows, horizons, and branch comparisons are frozen before evaluation.

Several audit principles follow from the theory.
\begin{enumerate}[label=(\arabic*),leftmargin=0.8cm]
\item \emph{Separate modal and nonmodal claims.}  Report $\rho(M_t)$ and finite-horizon gains independently.
\item \emph{Specify the metric.}  Repeat nonmodal conclusions under physically meaningful and standardized metrics.
\item \emph{Condition on stability when comparing stable response geometry.}  A gain comparison should not be driven solely by one branch crossing the unit circle.
\item \emph{Test nested observed spaces.}  A scalar may detect susceptibility while a two- or three-dimensional space is needed for closure.
\item \emph{Stress nonstationarity explicitly.}  Affine forcing and an intercept already span linear time on a ramp; quadratic time is a nontrivial additional stress.
\item \emph{Separate description from warning skill.}  Co-movement before a transition is not a validated EWS without chronological out-of-sample testing or independent trajectories.
\end{enumerate}

\section{AMOC hysteresis: a mechanistic application}
\label{sec:amoc}

\subsection{Experiment and frozen representation}

We use the $4400$-year Community Earth System Model freshwater-hosing experiment introduced by \citet{vanWesten2024}.  North Atlantic freshwater forcing increases linearly from $0$ to $0.66$ Sv over years $1$--$2200$ and then decreases symmetrically over years $2201$--$4400$.  The AMOC collapses near model year $1758$.  In the extended descending branch used here, physical recovery begins near year $3984$ and the fastest recovery occurs approximately during years $4114$--$4143$.

The full local state is a six-dimensional smooth representation of the AMOC overturning profile at $26^\circ$N between $500$ and $4500$ m.  A frozen control-period EOF analysis defines an ordered resolved space.  The first three EOFs explain $88.73\%$, $9.33\%$, and $1.38\%$ of control variance, respectively, for a cumulative $99.446\%$.  The three directions are well separated in variance and almost entirely contained in the smooth six-dimensional space (principal cosines $0.99996$, $0.99970$, and $0.99744$).

Local operators are estimated every ten years with trailing windows $W\in\{100,150,200\}$ years.  The primary window is $W=100$ and the primary response horizon is $H=50$ years.  Comparisons pair ascending and descending windows at identical freshwater forcing.  Ridge parameters are selected only in control.  The primary branch contrasts use the common cohort for which both compared operators are Schur stable.

\subsection{Physical meaning of the EOF coordinates}

\begin{figure}[t]
\centering
\begin{subfigure}[t]{0.49\textwidth}
\centering
\includegraphics[width=\linewidth]{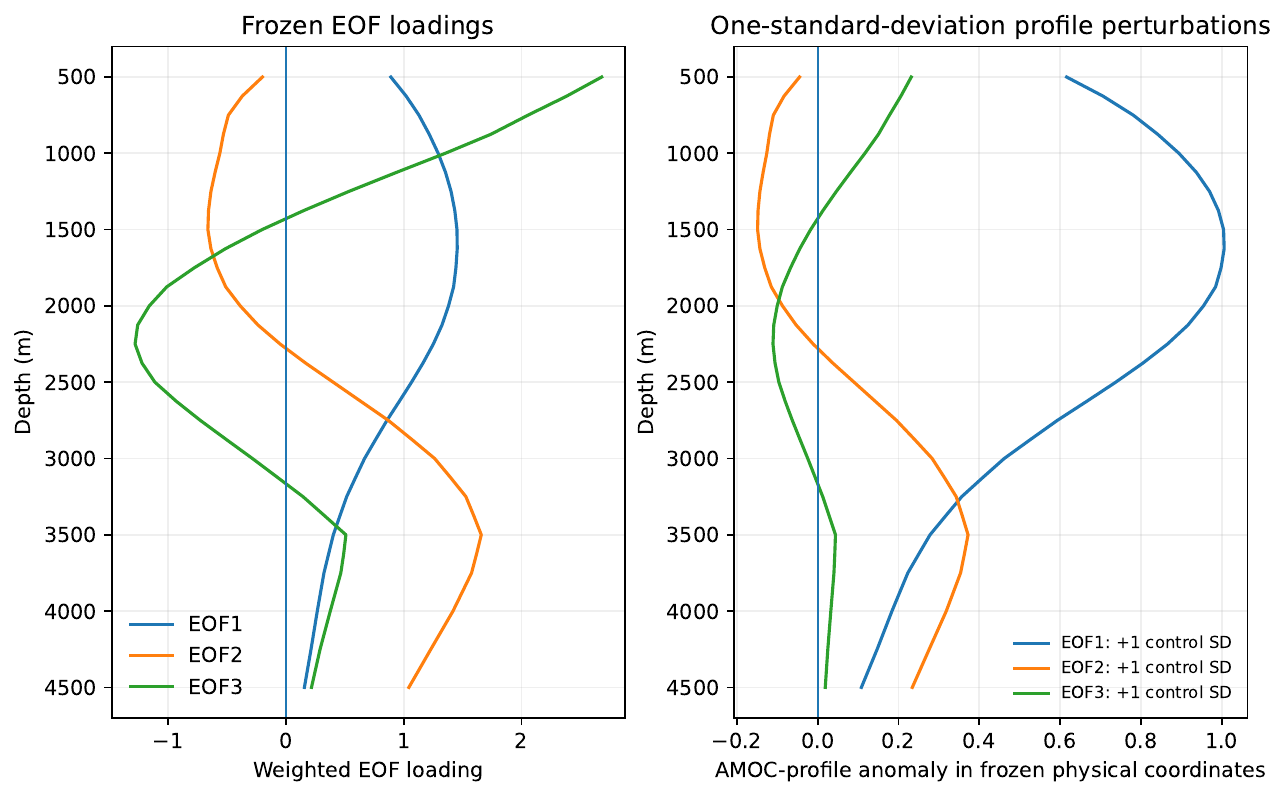}
\caption{Frozen weighted loadings and one-control-standard-deviation profile perturbations.}
\end{subfigure}\hfill
\begin{subfigure}[t]{0.49\textwidth}
\centering
\includegraphics[width=\linewidth]{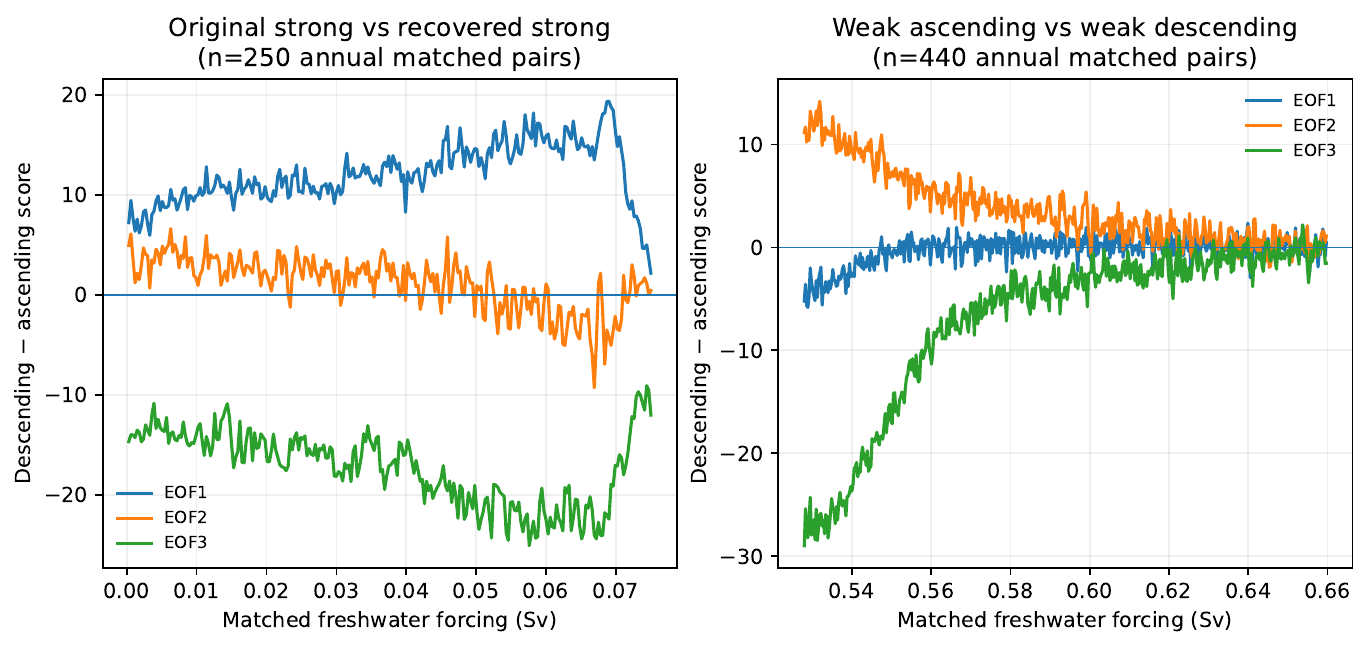}
\caption{Descending-minus-ascending scores at identical forcing.}
\end{subfigure}
\caption{Physical interpretation of the resolved coordinates.  EOF1 is an amplitude/transport mode, EOF2 a vertical-displacement and upper--deep redistribution mode, and EOF3 a shape/curvature mode.  In the weak regime, EOF1 is nearly identical across branches while EOF2 and EOF3 differ, providing a concrete failure of scalar state equivalence.}
\label{fig:eof-physics}
\end{figure}

EOF1 is an amplitude mode: its score has correlation $0.994$ with the upper-cell maximum after forcing adjustment.  EOF2 is dipolar and is most strongly associated with the depth of the upper-cell maximum (adjusted correlation $-0.888$).  EOF3 is tripolar, has its closest geometric match to a curvature template, and records thickness and upper--intermediate--deep compensation.

At matched forcing, the median recovered-strong minus original-strong score differences are
\[
\Delta a_1=11.89,
\qquad
\Delta a_2=1.72,
\qquad
\Delta a_3=-16.72
\]
in frozen control-standard-deviation units.  More sharply, for weak ascending versus weak descending states,
\[
\Delta a_1=-0.061,
\qquad
\Delta a_2=2.93,
\qquad
\Delta a_3=-3.58.
\]
Thus two weak states can have essentially the same scalar amplitude coordinate and different vertical geometry.

\subsection{Nested observed spaces}

We repartition the same frozen six-dimensional operators without refitting, using
\[
P_1=\operatorname{span}\{\mathrm{EOF1}\},
\quad
P_2=\operatorname{span}\{\mathrm{EOF1},\mathrm{EOF2}\},
\quad
P_3=\operatorname{span}\{\mathrm{EOF1},\mathrm{EOF2},\mathrm{EOF3}\}.
\]
Table~\ref{tab:nested} reports the recovered-strong minus original-strong median differences for the primary $W=100$, $H=50$ common-stable cohort ($n=13$).

\begin{table}[t]
\centering
\caption{Nested observed-space contrasts after transport recovery.  ``Physical'' is the depth-weighted metric and ``standardized'' the control-standardized Euclidean metric.  The same full operator is used at all ranks.}
\label{tab:nested}
\begin{tabular}{llrrr}
\toprule
Metric & Diagnostic & EOF1 & EOF1--2 & EOF1--3 \\
\midrule
Physical & hidden$\to$observed gain & 4.748 & 5.308 & 5.263 \\
Physical & projected-memory correction & $-0.006$ & 0.460 & 0.525 \\
Physical & static Schur return & $-0.150$ & 0.430 & 1.039 \\
\addlinespace
Standardized & hidden$\to$observed gain & 1.364 & 1.402 & 1.386 \\
Standardized & projected-memory correction & 0.055 & 0.379 & 0.406 \\
Standardized & static Schur return & $-0.051$ & 0.252 & 0.532 \\
\bottomrule
\end{tabular}
\end{table}

The hidden-to-observed gain contrast is positive for all ranks and all $13$ pairs.  EOF1 alone therefore detects altered susceptibility.  It does not robustly identify the changed feedback geometry: the projected-memory and static Schur-return contrasts become consistently positive only after EOF2 is included.  EOF1--EOF2 is the smallest physically interpretable space that resolves both amplitude and vertical displacement.  EOF3 enlarges the observed expression of memory and separates hysteresis branches, but its direct channel into EOF1--EOF2 is not robustly amplified under both metrics.  The dominant metric-robust pathway is residual vertical structure into EOF1 and EOF2.

\subsection{Recovered transport is not recovered operator geometry}

Under the four primary contracts---baseline and quadratic-time stress, each under physical and standardized metrics---the recovered strong branch has larger total gain, modal excess, hidden-to-observed gain, and projected-memory correction.  For $W=100$, $H=50$, and the $13$ common-stable pairs, the median recovered-minus-original differences are shown in \cref{tab:operator-contrast}.

\begin{table}[t]
\centering
\caption{Recovered strong minus original strong operator contrasts.  Every row has positive direction in all $13$ matched pairs under each contract.}
\label{tab:operator-contrast}
\begin{tabular}{lrrrr}
\toprule
Diagnostic & Base/physical & Base/std. & Quad./physical & Quad./std. \\
\midrule
Absolute modal excess & 5.189 & 1.134 & 4.782 & 1.254 \\
Full maximum gain & 5.194 & 1.167 & 4.831 & 1.181 \\
Hidden-to-observed gain & 5.263 & 1.386 & 4.858 & 1.371 \\
Projected-memory correction & 0.525 & 0.406 & 1.021 & 0.767 \\
\bottomrule
\end{tabular}
\end{table}

The magnitude of gain is metric dependent, as required by \cref{thm:metric}; the sign of the recovered-branch contrast is not.  The result is therefore not that the recovered state is spectrally unstable.  It is that similar transport can coexist with a different stable input--output and memory geometry.

\subsection{The \texorpdfstring{$\FovS$}{FovS} turning point and the operator}

The AMOC-induced freshwater transport at the southern Atlantic boundary is a physics-based scalar tied to the salt-advection feedback \citep{Hawkins2011,Jackson2013,vanWesten2024}.  The local archive contains $\FovS$ for the ascending branch only.  An ensemble of natural cubic splines built from all offsets of $50$-year means reproduces the published minimum exactly:
\[
t_{\min}^{\FovS}=1732,
\qquad
\FovS(t_{\min})=-0.134896\ \mathrm{Sv},
\]
which is $26$ years before the AMOC collapse at year $1758$.

\begin{figure}[t]
\centering
\begin{subfigure}[t]{0.49\textwidth}
\centering
\includegraphics[width=\linewidth]{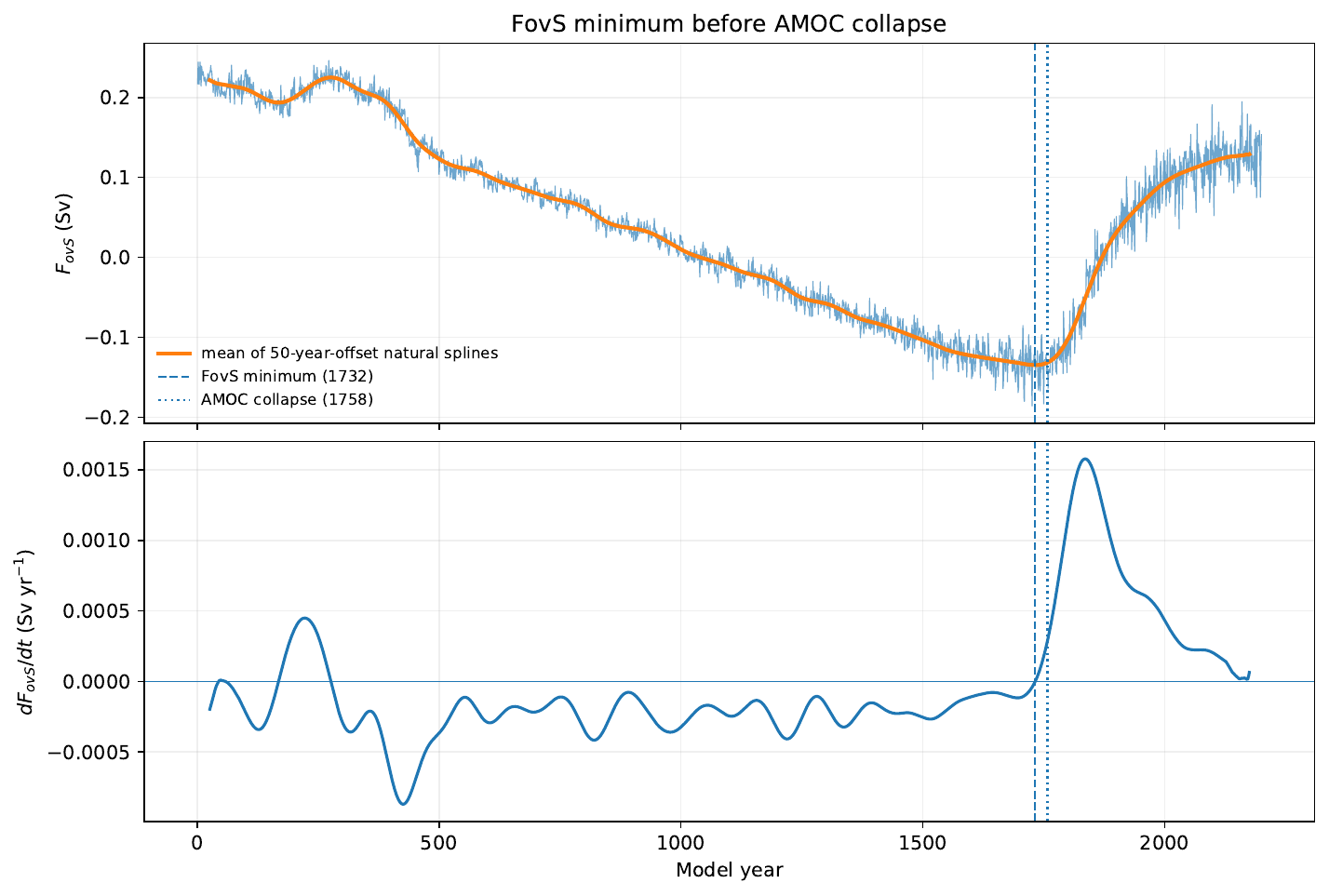}
\caption{$\FovS$ and the derivative of the spline ensemble.}
\end{subfigure}\hfill
\begin{subfigure}[t]{0.49\textwidth}
\centering
\includegraphics[width=\linewidth]{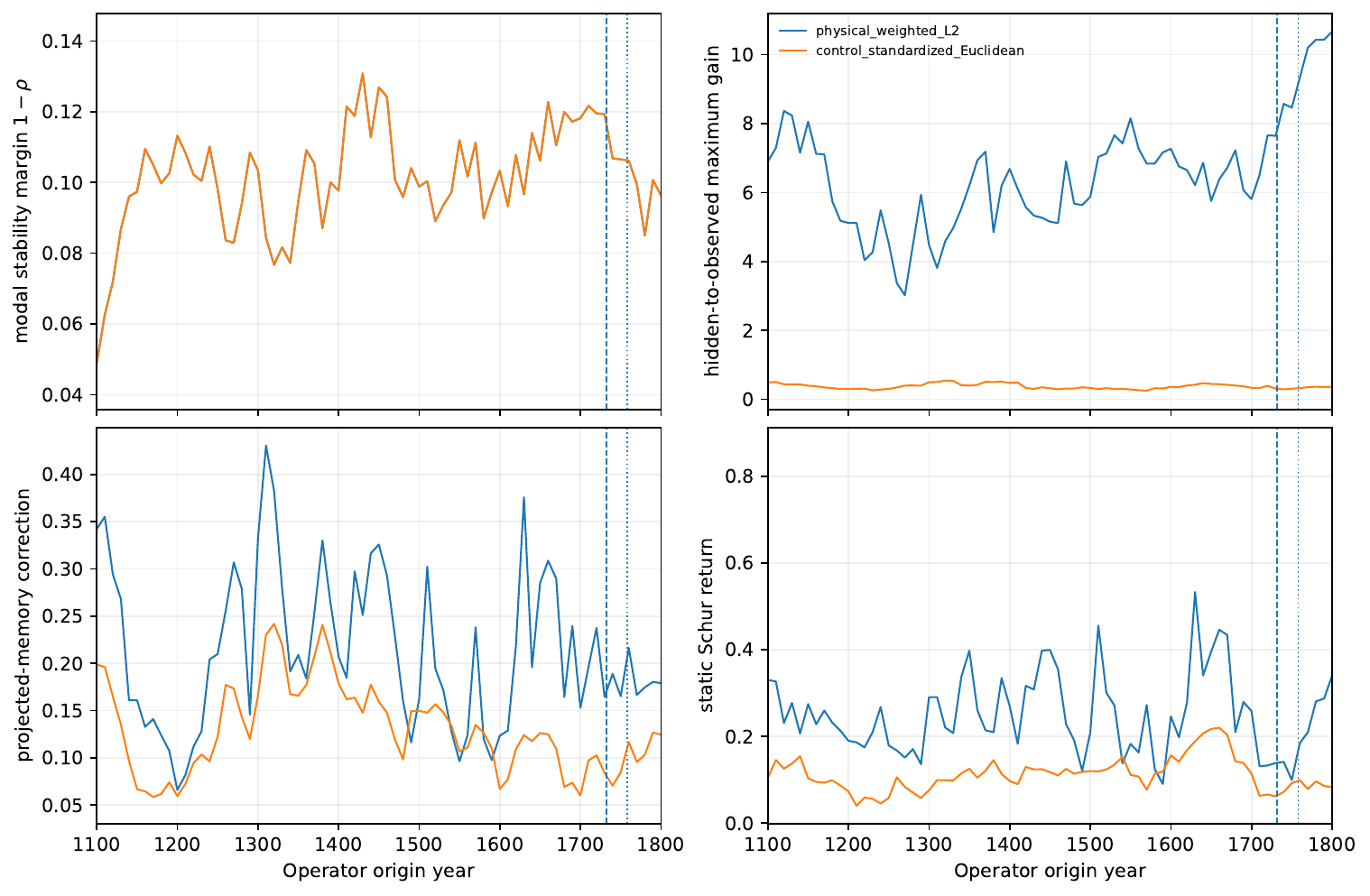}
\caption{Local operator diagnostics before and through collapse.}
\end{subfigure}
\caption{A physics-based scalar turning point and projected operator geometry need not share a threshold.  The $\FovS$ minimum occurs in year $1732$, while the modal stability margin does not decrease consistently toward zero.  Hidden-to-observed gain rises strongly around the collapse under the physical metric, but its pre-minimum trend is not metric invariant.}
\label{fig:fovs-operator}
\end{figure}

Five shifted $300$-year pre-collapse intervals were used to test interval sensitivity.  The variance of $\FovS$ increased in $5/5$ intervals, whereas its AC(1) increased in $3/5$.  For EOF1, variance and AC(1) increased in $3/5$ intervals each.  This reproduces the qualitative conclusion that the variance of the physics-based indicator is more interval robust than classical scalar AMOC diagnostics \citep{vanWesten2024}.

The operator diagnostics do not yield a universal pre-collapse EWS.  Warning-direction consistency across the same five intervals is summarized in \cref{tab:ews-consistency}.

\begin{table}[t]
\centering
\caption{Fraction of shifted pre-collapse intervals with the prespecified warning direction, $W=100$.  For the modal margin the warning direction is decreasing; for all other diagnostics it is increasing.}
\label{tab:ews-consistency}
\begin{tabular}{lcc}
\toprule
Diagnostic & Physical metric & Standardized metric \\
\midrule
Modal stability margin $1-\rho$ & 1/5 & 1/5 \\
Hidden-to-observed maximum gain & 4/5 & 4/5 \\
Projected-memory correction & 3/5 & 2/5 \\
Static Schur return & 3/5 & 4/5 \\
\bottomrule
\end{tabular}
\end{table}

Over years $1500$--$1731$, the Spearman correlation of the modal stability margin with $\FovS$ is $-0.733$, but the margin has a positive time slope and reaches its local minimum near year $1520$.  Hence the spectrum moves, on average, farther from the unit circle while $\FovS$ approaches its minimum.  Hidden-to-observed gain has opposite pre-minimum trends under the physical and standardized metrics; projected memory also reverses sign across metrics; static Schur return has near-zero rank correlation with the $\FovS$ level.  These results empirically instantiate \cref{prop:turning-point}:
\[
\frac{\dd\FovS}{\dd t}=0
\quad\not\Longleftrightarrow\quad
\rho(M_t)=1
\quad\not\Longleftrightarrow\quad
\det\Phi_t(1)=0.
\]

The correct interpretation is complementary.  $\FovS$ diagnoses the evolution of a salt-advection balance.  The Schur--Volterra operator diagnoses vertical input--output geometry, projected memory, and transient receptivity.  Some nonmodal quantities show warning-directed tendencies, but they are interval- or metric-dependent and are therefore treated as mechanistic diagnostics rather than validated EWS.

\subsection{Predictive scope}

A strict-causal decoder reconstructs the hidden coordinates from the previous $20$ years of resolved history and prescribed forcing.  At five-year horizon during recovery, the $W=100$ decoded model improves on a prespecified finite-lag-20 benchmark with skill $0.252$ and a descriptive block interval $[0.103,0.541]$; the full-state oracle skill is $0.308$.  The decoder captures $81.7\%$ of the oracle advantage.  The advantage survives deletion of any individual origin and any contiguous block of four origins, but does not robustly beat persistence.  This is consistent with \cref{thm:decoder-oracle}: hidden reconstruction is useful because $PJ^hQ$ is large, yet the observable remains highly persistent.

\subsection{Limitations of the application}

The local windows overlap and arise from a single forced trajectory.  Bootstrap intervals are dependence-aware descriptive summaries, not independent-replicate confidence intervals.  The descending branch lacks $\FovS$, so no statement is made about recovery of the salt-advection feedback.  The application supports the distinction
\[
\text{recovered transport}
\ne
\text{recovered operator geometry},
\]
not an observational collapse date or a universal climate tipping criterion.

\section{Discussion}
\label{sec:discussion}

\subsection{Detectability is not identifiability}

The scalar EWS theorem establishes detectability: under explicit conditions, the scalar inherits a singular covariance from a critical mode.  It does not establish identification of the full operator.  Different hidden realizations can generate the same transfer function, and different long kernels can agree over any finite horizon.  A scalar can therefore indicate increasing susceptibility while remaining insufficient to determine whether the mechanism is spectral criticality, nonnormal transfer, or hidden memory.

The AMOC nested-space audit makes this distinction concrete.  EOF1 detects increased hidden-to-observed susceptibility after recovery.  Only after EOF2 is included do the static Schur return and projected-memory contrast become robust.  In physical terms, intensity is detectable from a scalar, but resilience geometry requires intensity plus vertical displacement.

\subsection{Three memories and two kinds of warning}

Realized memory, structural return, and finite-horizon memory answer different questions.  A pathwise memory share can be high because recent events activate a fixed kernel.  A structural Schur return can be large even on a quiet trajectory.  A projected-memory correction can be large at one horizon and small at another because hidden excursions interfere.  None is automatically an EWS.

Likewise, a physics-based warning and a critical-slowing-down warning need not coincide.  The former may arise from a balance law or feedback reversal; the latter from a multiplier approaching a stability boundary.  Equivalence requires additional assumptions linking the scalar balance to the critical eigendirection and the Schur complement.  The $\FovS$ minimum provides a particularly clear example of a useful scalar threshold that is not the same event as a local spectral crossing of the vertical operator.

\subsection{Nonnormality and metric choice}

Nonnormality complicates both positive and negative conclusions.  Rising variance can occur at fixed spectrum, as \cref{thm:nonnormal-counterexample} proves.  Conversely, a physically meaningful gain can appear large in one metric and modest in standardized coordinates.  This is not a contradiction: the metric specifies which perturbations and responses are considered equally large.  A robust empirical claim should therefore distinguish:
\[
\text{spectral invariants},
\qquad
\text{metric-dependent gains},
\qquad
\text{metric-robust directions of contrast}.
\]

\subsection{Relation to fractional and event-time memory}

Fractional kernels belong naturally to the lifting framework through completely monotone relaxation spectra, but they are one model class among many.  Positive scalar event-time memory, including Hawkes-type excitation, is another special case.  The common principle is that statistics computed after projection need not identify the organization of the hidden or event-time dynamics.  The present theory extends that principle to signed, multivariate, nonnormal, and operator-valued memory.

\subsection{What would constitute a validated operatorial EWS?}

A stronger claim would require more than the present single-trajectory alignment.  At minimum one would need prespecified diagnostics, multiple independent or perturbed transition trajectories, chronological evaluation of lead time and false alarms, metric and representation stability, and comparison with physics-based and parsimonious scalar baselines.  The theory in \cref{sec:ews-theory} supplies the conditions to test; the current AMOC application supplies a mechanistic example, not the final validation.

\section{Conclusion}

Projection converts hidden dynamics into Volterra memory.  A Markovian lift realizes a selected memory kernel but does not invert the projection.  The Schur resolvent determines how hidden relaxation renormalizes stability, while $PJ^hQ$ and $PJ^hP-A^h$ quantify finite-horizon hidden transfer and projected memory.  Scalar critical slowing down survives this reduction only when the critical mode is spectrally isolated, well conditioned, observed, excited, and sampled under quasi-stationary forcing.

The failure of any one condition changes the interpretation.  A scalar can miss a true instability, display increasing variance far from instability, or exhibit a physically meaningful turning point unrelated to a spectral boundary.  In the AMOC experiment, $\FovS$ anticipates collapse through a salt-advection balance, whereas the vertical operator diagnoses altered memory and receptivity, especially after transport recovery.  Their complementarity is the result: a scalar warning and an operatorial resilience diagnostic need not be redundant, and neither can be substituted for the other without additional structure.

\appendix

\section{Stationary covariance details}

\begin{lemma}[Uniform remainder in \cref{thm:scalar-ews}]
Under \cref{ass:critical-mode}, the remainder $R_\eps$ in \eqref{eq:cov-asymptotic} satisfies
\[
\norm{R_\eps}
\le
\frac{2C_1}{1-\gamma}+
\frac{C_2}{1-\gamma^2}
\]
for constants independent of $\eps$.
\end{lemma}

\begin{proof}
Expand the covariance series into projector--projector, two projector--$N$ cross terms, and the $N$--$N$ term.  Uniform boundedness of $v_\eps,w_\eps,\Sigma_\eps$ bounds the cross terms by $C_1\sum_{k\ge0}(r_\eps\gamma)^k\le C_1/(1-\gamma)$ and the final term by $C_2\sum_{k\ge0}\gamma^{2k}$.
\end{proof}

\section{A continuous-time analogue}

Let $\dd x=A_\eps x\dd t+G_\eps\dd W_t$, with a simple real eigenvalue $-\alpha_\eps\uparrow0$ and a uniformly stable complement.  Under the continuous counterparts of \cref{ass:critical-mode}, the stationary covariance satisfies
\[
\Gamma_\eps
=
\frac{q_\eps}{2\alpha_\eps}v_\eps v_\eps^*+O(1),
\]
and an observed scalar variance diverges like $\alpha_\eps^{-1}$.  The normalized autocovariance at fixed lag $\tau$ tends $e^{-\alpha_\eps\tau}\to1$.  The proof is identical using
\[
\Gamma_\eps=\int_0^\infty e^{tA_\eps}G_\eps G_\eps^*e^{tA_\eps^*}\dd t.
\]

\section{Reproducibility and data scope}

The AMOC numbers reported in \cref{sec:amoc} come from a frozen audit pipeline.  The representation, control basis, ridge parameters, local windows, operator horizon, metric contracts, stability cohort, and comparison regimes were fixed before the interpretive audits.  The EOF physical atlas used no model fitting.  The nested observed-space audit repartitioned the same fitted six-dimensional operators.  The $\FovS$ audit used only the ascending branch because all $2200$ descending entries in the local contract are missing; no values were imputed or reconstructed.  The spline minimum and the operator/EWS alignment were then computed from frozen outputs.

\section*{Declarations}

\paragraph{Funding.}
No funds, grants, or other support were received specifically for this work.

\paragraph{Competing interests.}
The author has no relevant financial or non-financial interests to disclose.

\paragraph{Author contributions.}
Mauricio Herrera-Mar\'in conceived the study, developed the mathematical theory, designed and implemented the computational audits, analyzed the results, and wrote and revised the manuscript.

\paragraph{Data and code availability.}
The mathematical results are self-contained.  The frozen scripts, derived tables, audit records, and figure-generating outputs supporting the AMOC application are archived in the versioned AMOC--MZ reproducibility repository on Zenodo, DOI: \href{https://doi.org/10.5281/zenodo.21606982}{10.5281/zenodo.21606982}.  The underlying CESM experiment and secondary physical diagnostics are described and cited in \cref{sec:amoc}.  No descending-branch values of $\FovS$ were imputed or reconstructed.

\paragraph{Ethics approval and consent.}
Not applicable.

\printbibliography[title={References}]

\end{document}